\documentclass[10pt,reqno]{amsart} 
\usepackage{amsmath,amssymb}
\usepackage{graphicx, color}
\usepackage[utf8]{inputenc}

\theoremstyle{plain}
\newtheorem{theorem}{Theorem}[section]
\newtheorem{proposition}[theorem]{Proposition}
\newtheorem{lemma}[theorem]{Lemma}

\theoremstyle{definition}

\theoremstyle{remark}
\newtheorem{remark}[theorem]{Remark}

\numberwithin{equation}{section}

\newcommand{\mc}[1]{{\mathcal #1}}

\newcommand{\bb}[1]{{\mathbb #1}}

\newcommand{\rme}{\mathrm{e}}

\newcommand{\rmd}{\mathrm{d}}
\newcommand{\eps}{\varepsilon}

\title[Particle approximation of the BGK equation]{Particle approximation of the BGK equation}

\author[P.\ Butt\`a]{Paolo Butt\`a}
\address{Paolo Butt\`a\hfill\break \indent
Dipartimento di Matematica, 
Sapienza Universit\`a di Roma,
\hfill\break \indent
P.le Aldo Moro 5, 00185 Roma, Italy}
\email{butta@mat.uniroma1.it}
\author[M.\ Hauray]{Maxime Hauray}
\address{Maxime Hauray\hfill\break \indent
I2M Universit\'e Publique de France
\hfill\break\indent
39, rue F. Joliot Curie, 13453 Marseille Cedex 13 
}
\email{maxime.hauray@univ-amu.fr}
\author[M.\ Pulvirenti]{Mario Pulvirenti}
\address{Mario Pulvirenti\hfill\break \indent
Dipartimento di Matematica, 
Sapienza Universit\`a di Roma,
\hfill\break \indent
P.le Aldo Moro 5, 00185 Roma, Italy}
\email{pulviren@mat.uniroma1.it}
\begin{document}

\begin{abstract}
In this paper we prove the convergence of a suitable particle system towards the BGK model. More precisely, we consider an interacting stochastic particle system in which each particle can instantaneously thermalize locally. We show that, under a suitable scaling limit, propagation of chaos does hold and the one-particle distribution function converges to the solution of the BGK equation. 
\end{abstract}



\maketitle
\thispagestyle{empty}

\section{Introduction}
\label{sec:1}

The BGK model is a kinetic equation of the form,
\begin{equation}
\label{BGK}
(\partial_t f+ v\cdot \nabla_x f) (x,v,t) = \lambda \big(\varrho(x,t) M_f (x,v,t)-f(x,v,t)\big)\,,
\end{equation}
where
\[
M_f(x,v,t) = \frac 1 {(2 \pi T(x,t))^{d/2} } \exp\left(-\frac {|v-u(x,t)|^2}{2T(x,t)}\right)
\]
and
\[
\begin{split}
& \varrho (x,t) = \int\! \rmd v\, f(x,v,t)\,, \quad \varrho u (x,t) = \int\! \rmd v\, f(x,v,t) v\,, \\ & \varrho(u^2 +Td)(x,t)  = \int\! \rmd v\, f(x,v,t) |v|^2 \,.
\end{split}
\]

Eq.~\eqref {BGK}  governs the time evolution of  the one-particle distribution function $f=f(x,v,t)$, where $(x,v)$ denotes position and velocity of the particle and $t$ is the time. Here, $d=1,2,3$ is the dimension of the physical space. The BGK model describes the dynamics of a tagged particle which thermalizes instantaneously at Poisson random time of intensity $\lambda >0$. The Maxwellian $M_f$ has mean velocity and temperature given by $f$ itself.

This model was introduced by P.L.\ Bhatnagar, E.P.\ Gross, and M.\ Krook in \cite {BGK} as a simpler substitute to the fundamental and physically founded Boltzmann equation.  Clearly, the BGK model preserves local mass, momentum, and  energy, so that it shares many physical properties with the Boltzmann equation. Moreover, it satisfies the $H$-Theorem, and therefore it exhibits the usual hydrodynamic behavior in the limit of vanishing mean free path. 

Roughly speaking the BGK model was introduced to handle situations where the mean free path is very small (but positive) so that the hydrodynamic picture is inadequate. To fix the ideas, we consider a stochastic particle system like the DSMCM (the Bird Montecarlo method), thus equivalent to the Boltzmann evolution,  when the intensity of the interactions is very large and the free motion is finite. The BGK leading idea is that it is useless to compute in detail the very many interactions taking place locally, since we know a priori that the system is locally thermalizing. This means that we can replace the true dynamics with a jump process in which position and velocity of a given particle are instantaneously distributed according to the local equilibrium.  Inspired by these arguments, in this work, we present a stochastic system of $N$ interacting particles yielding the BGK equation in a suitable scaling limit. In this microscopic model, each particle moves freely up to some random instant in which it performs a random jump in position and velocity. The outgoing position and velocity are chosen according to a given distribution and a Maxwellian respectively, both determined by the actual particle configuration. See Section \ref{sec:2.2} for details.

In the limit $N\to \infty$, we expect that the one particle distribution function converges to the solution of the BGK equation, provided that, at initial time, the particles are independent (i.e., their distribution factorizes).

Obviously, the dynamics creates correlations because of the jump mechanism, which depend on the state of the full particle configuration. Note that the interaction is not binary in the present context, so that we do not use hierarchical techniques to obtain propagation of chaos. 

Actually, the convergence follows from the fact that the action on a given particle produced by any other particle is small (as in the mean-field limit), so that we can expect to recover the propagation of chaos in the limit $N\to \infty$.

We mention the recent work \cite{MW}, where the one dimensional homogeneous linear BGK equation has been obtained as a limit of a suitable particle process in which the thermalization is driven by the Kac's model. Therefore, the context and the approach are different from the ones of the present paper.

We note that in the original work \cite{BGK}, the jump rate is chosen $\lambda = \varrho$, namely the jumps are favorite whenever the spatial density is high. This case is mathematically much more involved compared with the case in which the rate $\lambda$ is constant so that here, we assume for simplicity $\lambda=1$. From a mathematical point of view, a constructive existence and uniqueness theorem for the BGK equation was given in \cite{PP}. Previous non-constructive existence results, in the spirit of the Di Perna-Lions theorem for the Boltzmann equation, were obtained in \cite{P} (see also \cite{GP}). Regarding the hydrodynamic limit we mention, e.g., \cite{S}. Actually, we are not aware of any constructive existence theorem for the solutions of the BGK equation when $\lambda = \varrho$ and this makes difficult to approach the particle approximation problem. However, the particle system yielding, at least formally, this BGK kinetic equation makes perfectly sense as we shall discuss briefly in Section \ref{sec:5}.

The plan of the paper is the following. The next section is devoted to notation, preliminary material, and statement of the results. The remaining sections are devoted to the proofs. More precisely, the convergence follows from two separate results: the convergence of a particle dynamics towards a regularized version of the BGK equation, and the removal of the cut-off to recover the true BGK equation from its regularized version. The former requires the main effort and it is the content of Section \ref{sec:3}, while the latter is proven in Section \ref{sec:4}. Finally, some concluding remarks are given in Section \ref{sec:5}.

\section{Preliminaries and statement of the results}
\label{sec:2}

Let $\bb T^d = \big(\bb R/(\frac 12 + \bb Z)\big)^d$ be the $d$-dimensional torus of side length one. We denote by $M_{u,T}=M_{u,T}(v)$, $v\in\bb R^d$, the normalized Maxwellian density of mean velocity $u\in \bb R^d$ and temperature $T$, i.e.,
\begin{equation}
\label{max}
M_{u,T}(v) = \frac1{(2\pi T)^{d/2}} \exp\left(-\frac{|v-u|^2}{2T}\right).
\end{equation}
In particular,
\[
u = \int\!\rmd v\, M_{u,T}(v)\, v \,, \qquad T = \frac 1d \int\!\rmd v\, M_{u,T}(v)\, |v-u|^2\,.
\]

\subsection{The BGK equation and its regularized version}
\label{sec:2.1}
We denote by $f = f(t) = f(x,v,t)$, where $(x,v)\in \bb T^d \times \bb R^d$ and $t\in \bb R_+$ is the time, the solution to the BGK equation,
\begin{equation}
\label{eq:bgk}
\partial_t f + v\cdot \nabla_x f = \varrho_f M_f - f\,,
\end{equation}
where $\varrho_f= \varrho_f(x,t)$ is the local density defined by
\begin{equation}
\label{eq:rho}
\varrho_f(x,t) = \int\!\rmd v\, f(x,v,t)\,,
\end{equation}
while $M_f = M_f(x,v,t)$ is the (local) Maxwellian given by
\begin{equation}
\label{eq:maxf}
M_f(x,v,t) = M_{u_f(x,t),T_f(x,t)}(v)\,,
\end{equation}
where $u_f=u_f(x,t)$ and $T_f=T_f(x,t)$ are the local velocity and temperature, 
\begin{align}
\label{eq:uf}
\varrho_f(x,t) u_f(x,t) & = \int\! \rmd v\, f(x,v,t)\, v\,, \\ \label{eq:tf} \varrho_f(x,t) T_f(x,t) & = \frac 1d \int\! \rmd v\, f(x,v,t)\, |v-u_f(x,t)|^2\,.
\end{align}

We also consider the solution $g = g(t) =g(x,v,t)$ of the ``regularized'' BGK equation,
\begin{equation}
\label{eq:kin}
\partial_t g + v\cdot \nabla_x g = \varrho_g^\varphi M_g^\varphi - g\,,
\end{equation}
where
\begin{equation}
\label{eq:maxg}
M_g^\varphi(x,v,t) = M_{u_g^\varphi(x,t),T_g^\varphi(x,t)}(v)\,.
\end{equation}
In Eqs.~\eqref{eq:kin} and \eqref{eq:maxg}, $\varrho_g^\varphi = \varrho_g^\varphi(x,t) $, $u_g^\varphi=u_g^\varphi(x,t)$, and $T_g^\varphi=T_g^\varphi(x,t)$ are smeared versions of the local density, velocity, and temperature. More precisely,
\begin{align}
\label{eq:rphi}
\varrho_g^\varphi(x,t) & = (\varphi*\varrho_g) (x,t) = \int\!\rmd y\, \varphi(x-y)\varrho_g(y,t), \\ \label{eq:uphi}
\varrho_g^\varphi(x,t) u_g^\varphi (x,t) & = \int\! \rmd y\, \rmd v\, \varphi(x-y) g(y,v,t)\, v\,, \\ \label{eq:tphi} \varrho_g^\varphi(x,t) T_g^\varphi (x,t) & = \frac 1d \int\! \rmd y\, \rmd v\, \varphi(x-y) g(y,v,t)\, |v-u_g^\varphi(x,t)|^2\,,
\end{align}
where
\begin{equation}
\label{eq:rhog}
\varrho_g(x,t) = \int\!\rmd v\, g(x,v,t)
\end{equation}
and $\varphi$ is a strictly positive, even, and smooth smearing function, i.e.,
\begin{equation}
\label{eq:varphi}
\varphi\in C^\infty(\bb T^d;\bb R_+)\,,  \qquad \varphi(x) = \varphi(-x)\,, \qquad \int\!\rmd y\, \varphi(y) = 1\,.
\end{equation}

Well-posedness of the BGK equation together with $L^\infty$ estimates for the hydrodynamical fields can be found in \cite{PP}. In particular, the following proposition follows immediately from \cite[Theorem 3.1]{PP}.

\begin{proposition}
\label{prop:bgk}
Let $f_0$ be a probability density on $\bb T^d\times \bb R^d$ and suppose there are a function $a\in C(\bb R^d)$ and positive constants $C_1,\alpha>0$ such that
\begin{equation}
\label{eq:f0}
\begin{split}
& a(v) \le f_0(x,v) \le C_1\rme^{-\alpha|v|^2} \quad \forall\, (x,v)\in \bb T^d\times \bb R^d\,, \\ & a \ge 0\,, \quad C_2 := \int\!\rmd v \, a(v) >0\,.
\end{split}
\end{equation}
Then there exists a mild solution $f=f(t)=f(x,v,t)$ to Eq.~\eqref{eq:bgk} with initial condition $f(x,v,0) = f_0(x,v)$.\footnote{This means that $f$ solves the integral equation,
\[
f(x,v,t) = \rme^{-t} f_0(x-vt,v) + \int_0^t\!\rmd s\, \rme^{-(t-s)} (\varrho_f M_f) (x-v(t-s),v,s)\,,
\] 
which formally derives from Eq.~\eqref{eq:bgk} via Duhamel formula.} Moreover, there are a non-decreasing finite function $t\mapsto K_{q,t} = K_{q,t}(f_0)$, $q\in\bb N$, and a non-increasing positive function $t\mapsto A_t = A_t(f_0)$ such that, for any $(x,t)\in \bb T^d\times \bb R_+$,
\begin{align}
\label{eq:utf}
& |u_f(x,t)| + T_f(x,t) + \mc N_q(f(t)) \le K_{q,t}\,, \\ \label {eq:rT>} & \varrho_f(x,t) \ge C_2 \rme^{-t}\,, \quad T_f (x,t) \ge A_t \,,
\end{align}
where
\begin{equation}
\label{eq:Nn}
\mc N_q(f) := \sup_{(x,v)\in \bb T^d\times \bb R^d} f(x,v) (1+ |v|^q)\,.
\end{equation}
Finally, the above solution is unique in the class of functions $f=f(t)=f(x,v,t)$ such that, for some $q>d+2$, $\sup_{t\le \tau}\mc N_q(f(t)) < +\infty$ for any $\tau>0$.
\end{proposition}

The analysis in \cite{PP} can be extended straightforwardly to the regularized BGK equation, in particular the $L^\infty$ estimates do not depend on the smearing function $\varphi$. This is the content of the following proposition, whose proof is sketched in Appendix~\ref{app:a}.
\begin{proposition}
\label{prop:stim_uT}
Let $g=g(t)=g(x,v,t)$ be the solution to Eq.~\eqref{eq:kin} with initial condition $g(x,v,0) = f_0(x,v)$, $f_0$ as in Proposition \ref{prop:bgk}. Then, similar estimates hold for the corresponding hydrodynamical fields, namely,
\begin{align}
\label{eq:utK}
& |u_g^\varphi(x,t)| + T_g^\varphi(x,t) + \mc N_q(g(t)) \le K_{q,t}\,, \\ \label{stimrho} & \varrho_g(x,t) \ge C_2 \rme^{-t}\,, \quad \varrho_g^\varphi(x,t) \ge C_2\rme^{-t}\,, \\ \label{eq:TA} & T_g^\varphi (x,t) \ge A_t\,,
\end{align}
(with $K_ {q,t}$, $A_t$ independent of $\varphi$). 
\end{proposition}

\subsection{The stochastic particle system}
\label{sec:2.2}

We consider a system of $N$ particles moving in the $d$-dimensional torus $\bb T^d$. We denote by $Z_N=(X_N,V_N)$ the state of the system, where $X_N\in(\bb T^d)^N $ and $V_N\in (\bb R^d)^N $ are the positions and velocities of particles, respectively.

Recalling $\varphi$ denotes a smearing localizing function with the properties detailed in Eq.~\eqref{eq:varphi}, setting $X_N=(x_1, \dots, x_N)$ and $V_N=(v_1, \dots, v_N)$,
we introduce the (smeared) empirical hydrodynamical fields $\varrho_N^\varphi$, $u_N^\varphi$, and $T_N^\varphi$ (depending on $Z_N$) defined by \[
\begin{split}
& \varrho_N^\varphi(x) = \frac 1N \sum_{j=1}^N \varphi (x-x_j)\,, \quad  \varrho_N^\varphi u_N^\varphi(x) = \frac 1N \sum_{j=1}^N \varphi (x-x_j) v_j\,, \\ & \varrho_N^\varphi T_N^\varphi(x) = \frac 1{Nd} \sum_{j=1}^N \varphi (x-x_j) |v_j-u_N^\varphi(x)|^2 \,.
\end{split}
\]

The system evolves according to a Markovian stochastic dynamics, whose generator $\mc L_N$ is defined as
\begin{align}
\label{gen}
\mc L_N{G}(Z_N) & = [(V_N\cdot\nabla_{X_N}- N)G] (Z_N) \nonumber \\ & \quad + \sum _{i=1}^N \int\! \rmd \tilde x_i \, \rmd \tilde v_i\, \varphi(\tilde x_i - x_i) M_{Z_N}^\varphi (\tilde x_i,\tilde v_i) {G} (Z_N^{i,(\tilde x_i,\tilde v_i)})\,,
\end{align}
where $Z_N^{i,(y,w)}=(X_N^{i,y},V_N^{i,w})$ is the state obtained from $Z_N=(X_N,V_N)$ by replacing the position $x_i$ and velocity $v_i$ of the $i$-th particle by $y$ and $w$ respectively; ${G}$ is a test function on the state space, and $M_{Z_N}^\varphi(x,v) $ is the Maxwellian constructed via the empirical fields,
\[
M_{Z_N}^\varphi(x,v) = M_{u_N^\varphi(x),T_N^\varphi(x)}(v)\,.
\]
We emphasize that the process is well defined since, as $\varphi$ is strictly positive, the smeared hydrodynamical temperature $T_N^\varphi(x,t)$ is vanishing only if $v_j=u_N^\varphi(x,t)$ for all $j=1,\ldots,N$, and this is a negligible event, and even if $T_N^\varphi(x,t)=0$, we could replace the Maxwellian by a Dirac mass in $u_N^\varphi(x,t)$, and the $N$ particle dynamics will be well defined in any case.

The generator Eq.~\eqref{gen} is associated to the process $Z_N(t) = (X_N(t),V_N(t))$ in which at each Poisson time, of intensity $N$, a particle chosen with probability $1/N$ performs a jump from its actual position and velocity $(x_i,v_i)$ to the new ones $(\tilde x_i,\tilde v_i)$, extracted according to the distribution $\varphi(\cdot-x_i)$ for the position and then to the empirical Maxwellian $M_{Z_N}^\varphi (\tilde x_i,\cdot)$ for the velocity. In the sequel, we will denote by $F_N(t) = F_N(Z_N,t)$ the density of the law of $Z_N(t)$ (but we will often refer to it as simply the law of the process).

\smallskip
\noindent\textbf{A notation warning.} In what follows, we shall denote by $C$ a generic positive constant whose numerical value may change from line to line and it may possibly depend on time $t$ and initial condition $f_0$.

\subsection{Results}
\label{sec:2.3}

The particle approximation of the BGK equation is achieved in two steps. We first show that, for fixed smearing function $\varphi$, the stochastic dynamics defined above is a particle approximation to the BGK regularized equation Eq.~\eqref{eq:kin}. This is in fact the main result of the paper and it is the content of Theorem \ref{teo:main} below. We next consider a $\delta$-approximating sequence $\{\varphi_\eps\}$ of smearing functions and show that the corresponding BGK regularized equations furnish an approximation to the BGK equation. From these results we deduce that the stochastic dynamics constructed with smearing function $\varphi = \varphi_{\eps_N}$, for a suitable choice of $\eps_N$ (converging to $0$ slowly as $N\to +\infty$), gives the required particle approximation of the BGK equation.

\begin{theorem}
\label{teo:main}
Suppose that the law of $Z_N(0)$ is $F_N(0) = f_0^{\otimes N}$, where $f_0$ satisfies the assumptions detailed in Eq.~\eqref{eq:f0}. Let $g=g(t)=g(x,v,t)$ be the solution to Eq.~\eqref{eq:kin} with initial condition $g(0)=f_0$ and smearing function $\varphi$ as detailed in Eq.~\eqref{eq:varphi}. Define also
\begin{equation}
\label{Cphi}
\Gamma_\varphi := (1+C_\varphi^8) (1+\|\varphi\|_\infty^8) (1+\|\nabla\varphi\|_\infty^2)\,, \qquad C_\varphi := \Big(\min_{x\in \bb T^d} \varphi(x)\Big)^{-1}\,.
\end{equation}
Let $f_j^N(t)$, $j\in\{1,\ldots, N\}$, be the $j$-particle marginal distribution function of the (symmetric) law $F_N(t)$,i.e.,
\[
f_j^N(x_1,\ldots,x_j,v_1,\ldots,v_j,t) = \int\! \rmd x_{j+1} \cdots \rmd x_N\, \rmd v_{j+1} \cdots \rmd v_N\, F_N(X_N,V_N,t)\,.
\]
Then, the 2-Wasserstein distance\footnote{If $\mu$ and $\nu$ are two probability measures on a metric space $(M,d)$ with finite second moment, the 2-Wasserstein distance between $\mu$ and $\nu$ is defined as 
\[
\mc W_2(\mu ,\nu) = \left(\inf_{\gamma\in \mc P(\mu,\nu)} \int_{M\times M}\!\rmd \gamma(x,x')\, d(x,x')^2\right)^{1/2},
\]
where $\mc P(\mu,\nu)$ denotes the collection of all measures on $M\times M$ with marginals $\mu$ and $\nu$. Here, $M=(\bb T^d)^j\times (\bb R^d)^j$ and $\mc W_2\big(f_j^N(t), g(t)^{\otimes j}\big)$ denotes the 2-Wasserstein distance between the probability measures with densities $f_j^N(t)$ and $g(t)^{\otimes j}$ respectively.} $\mc W_2\big(f_j^N(t), g(t)^{\otimes j}\big)$ vanishes as $N\to +\infty$ for any $j\in \bb N$ and $t \ge 0$. More precisely, there exists a non-decreasing finite function $t\mapsto L_t = L_t(f_0)$ such that, for any $j\in\{1,\ldots, N\}$ and $t\ge 0$,
\begin{equation}
\label{w2s}
\mc W_2\big(f_j^N(t),g(t)^{\otimes j}\big)^2 \le \frac jN L_t\Gamma_\varphi \exp (L_t\Gamma_\varphi)\,.
\end{equation}
In particular, the one particle marginal distribution function $f_1^N(t)$  weakly converges to $g(t)$ as $N\to +\infty$ for any $t \ge 0$.
\end{theorem}

Now, we fix a sequence $\{\varphi_\eps\}$, $\eps\in (0,1)$, of smearing functions such that, in addition to Eq.~\eqref{eq:varphi}, fulfil the following condition,
\begin{align}
\label{stim1}
& \|\varphi_\eps\|_\infty \le C \eps^{-d}\,, \qquad \|\nabla\varphi_\eps\|_\infty \le C \eps^{-(d+1)}\,, \qquad C_{\varphi_\eps} \le C\eps^{-1}\,, \\ \label{stim2} &  \|\varphi_\eps * J - J\|_\infty \le C(J) \eps\, \quad \forall\, J\in C^1(\bb T^d)\,,
\end{align}
with $C_{\varphi_\eps}$ as in Eq.~\eqref{Cphi} and $C(J)$ is a constant multiple of $(\|J\|_\infty + \|\nabla J\|_\infty)$. For example, given a smooth function $\Phi \colon \bb R^d \to \bb R_+$, with $\int\!\rmd z \, \Phi(z) =1$ and compactly supported inside the ball of radius $1/2$ centered in $z=0$, it is readily seen that the functions $\varphi_\eps$ on $\bb T^d$ defined by setting
\begin{equation}
\label{phieps}
\varphi_\eps(x) = \frac{\eps+\eps^{-d}\Phi(x/\eps)}{1+\eps}\,, \qquad  x\in \left(-\frac 12,\frac 12\right)^d, \;\; \eps\in (0,1)\,,
\end{equation}
satisfy the conditions in Eqs.~\eqref{eq:varphi}, \eqref{stim1}, and \eqref{stim2}.

We next denote by $g^\eps$ the solution to the regularized BGK equation with smearing function $\varphi_\eps$. Our goal is to compare $g^\eps$ with the solution $f$ of the BGK equation which satisfies the same initial condition. 

\begin{theorem}
\label{thm:fgd}
Assume $f(0) = g^\eps(0) = f_0$, where $f_0$ is a differentiable density satisfying the condition in Eq.~\eqref{eq:f0} and such that, for some $q>d+2$,
\begin{equation}
\label{dg0}
\mc N_q(|\nabla_x f_0|) < + \infty\,.
\end{equation}
Then, for any $t\ge 0$,
\begin{equation}
\label{eq:f-g}
\int\!\rmd x\, \rmd v\, (1+ |v|^2) \, |f(x,v,t) - g^\eps(x,v,t)| \le C \eps\,.
\end{equation}
\end{theorem}

It is now easy to construct the particle approximation to the BGK equation. First of all, we observe that, in view of Eq.~\eqref{Cphi}, $\Gamma_{\varphi_\eps} \le C \eps^{-\eta}$ with $\eta=10(d+1)$ for any $\eps$ small. Then, let $\tilde Z_N(t)$ be the process constructed with smearing function $\varphi= \varphi_{\eps_N}$, $\eps_N\to 0$ to be chosen. Suppose also that the law $\tilde F_N(t)$ of $\tilde Z_N(t)$ has initial value $\tilde F_N(0) = f_0^{\otimes N}$, where $f_0$ satisfies the assumptions of Theorem \ref{thm:fgd}. Let $f=f(t)=f(x,v,t)$ be the solution to Eq.~\eqref{eq:bgk} with initial condition $f(0)=f_0$. Then, from Theorem \ref{teo:main} and Theorem \ref{thm:fgd}, for any $N$ large enough we have,
\[
\begin{split}
\mc W_2\big(\tilde f_j^N(t),f(t)^{\otimes j}\big) & \le \mc W_2\big(\tilde f_j^N(t),g^{\eps_N}(t)^{\otimes j}\big) + \mc  W_2\big(g^{\eps_N}(t)^{\otimes j},f(t)^{\otimes j}\big) \\ & \le C\sqrt{\frac jN}\eps_N^{-\eta} \exp (C\eps_N^{-\eta}) + jC\eps_N\,.
\end{split}
\]
From this estimate we deduce the aimed result, which is the content of the following theorem.

\begin{theorem}
\label{teo:main1}
With the above notation, choose $\eps_N = (\log N)^{-1/\gamma}$ with $\gamma>\eta$. Then, letting $\tilde f_j^N(t)$, $j\in \{1, \ldots, N\}$, be the $j$-particle marginal distribution function of the (symmetric) law $\tilde F_N(t)$, the 2-Wasserstein distance $\mc W_2\big(\tilde f_j^N(t),f(t)^{\otimes j}\big)$ vanishes as $N\to +\infty$ for any $j\in \bb N$ and $t \ge 0$. In particular, the one particle marginal distribution function $\tilde f_1^N(t)$ weakly converges to $f(t)$ as $N\to +\infty$ for any $t \ge 0$.
\end{theorem}

\section{Particle approximation of the regularized BGK equation}
\label{sec:3}

In this section, we prove Theorem \ref{teo:main}. For reader convenience, the section is divided in several subsections corresponding to the different steps of the proof.

\subsection{Heuristics}
\label{sec:3.1}

Because of the mean field character of the interaction among the particles, we expect that the propagation of chaos property holds as the size $N$ of the system tends to infinity. We claim that if this is true then the one particle marginal distribution function $f_1^N = f_1^N(t) = f_1^N(x,v,t)$ of the law $F_N(t)$ does converge to the solution to Eq.~\eqref{eq:kin}. Indeed, from Eq.~\eqref{gen},
\[
\begin{split}
& \frac{\rmd}{\rmd t} \int\! \rmd x_1\, \rmd v_1\, f_1^N(x_1,v_1,t) \psi(x_1,v_1) = \frac{\rmd}{\rmd t} \int\! \rmd Z_N\, F_N(Z_N,t) \psi(x_1,v_1) \\ & \quad = \int\! \rmd x_1\, \rmd v_1\, f_1^N(x_1,v_1,t) (v_1 \cdot \nabla_{x_1}-1) \psi(x_1,v_1) \\ & \qquad + \int\! \rmd Z_N\, F_N(Z_N,t) \int\! \rmd \tilde x_1 \, \rmd \tilde v_1\, \varphi(\tilde x_1 - x_1) M_{Z_N}^\varphi (\tilde x_1,\tilde v_1) \psi (\tilde x_1, \tilde v_1 )\,,
\end{split}
\]
where $\psi$ is a test function  on the one-particle state space. Now, due to the law of large numbers, if $Z_N$ is distributed according to $F_N(t) \approx f_1^N(t)^{\otimes N}$ then
\[
\frac 1N \sum_i \delta (x-x_i)\delta(v-v_i) \approx f_1^N(x,v) \quad (\text{weakly})\,,
\]
whence
\[
M_{Z_N}^\varphi (\tilde x_1,\tilde v_1) \approx  M_{f_1^N}^\varphi (\tilde x_1,\tilde v_1)\,.
\]
Therefore,
\[
\begin{split}
& \int\! \rmd Z_N\, F_N(Z_N,t) \int\! \rmd \tilde x_1 \, \rmd \tilde v_1\, \varphi(\tilde x_1 - x_1) M_{Z_N}^\varphi (\tilde x_1,\tilde v_1) \psi (\tilde x_1, \tilde v_1 ) \\ & \qquad \approx \int\! \rmd x_1\, \rmd v_1\,  f_1^N(x_1,v_1,t) \int\! \rmd \tilde x_1 \, \rmd \tilde v_1\,\varphi(\tilde x_1 - x_1) M_{f_1^N}^\varphi (\tilde x_1,\tilde v_1) \psi (\tilde x_1, \tilde v_1) \\ & \qquad = \int\! \rmd \tilde x_1 \, \rmd \tilde v_1\, \varrho_{f_1^N}^\varphi(\tilde x_1)M_{f_1^N}^\varphi (\tilde x_1,\tilde v_1) \psi (\tilde x_1, \tilde v_1)\,, 
\end{split}
\]
and the claim follows.

Our purpose, Theorem \ref{teo:main}, is to prove rigorously this fact. This will be achieved by showing that the dynamics remains close to an auxiliary $N$-particle process, constituted by $N$ independent copies of the non-linear jump process associated to the kinetic equation.\footnote{This process is called non-linear since its generator is implicitly defined through the law of the process itself, see Eq.~\eqref{nonlin} further on.} The thesis of the theorem then follows by applying the law of large numbers to the auxiliary process.

\subsection{Coupling}
\label{sec:3.2}

The auxiliary process, named $\Sigma_N(t) = (Y_N(t),W_N(t)) \in (\bb T^d)^N\times (\bb R^d)^N$, is defined according to the following construction.

Let $g=g(t) = g(x,v,t)$ be as in Theorem \ref{teo:main} and denote by $(x(t),v(t))\in \bb T^d\times \bb R^d$ the one-particle jump process whose generator is given by
\begin{equation}
\label{nonlin}
\mc L_1^g\psi(x,v) = [(v\cdot \nabla_x -1)\psi](x,v) + \int\! \rmd \tilde x \, \rmd \tilde v\, \varphi(\tilde x - x) M_g^\varphi(\tilde x,\tilde v) \psi(\tilde x,\tilde v)\,,
\end{equation}
where $\psi$ is a test function and $M_g^\varphi$ is defined in Eq.~\eqref{eq:maxg}. We remark that if the initial distribution has a density then the same holds at positive time and for the probability density of $(x(t),v(t))$ solves the regularized BGK equation \eqref{eq:kin}.The auxiliary $N$-particle process $\Sigma_N(t)$ is then defined by $N$ independent copies of the above process. Otherwise stated, it is the Markovian dynamics on $(\bb T^d)^N\times (\bb R^d)^N$ with generator
\begin{align}
\label{genf}
\mc L_N^g{G}(Z_N) & = [(V_N\cdot\nabla_{X_N} -N){G}] (Z_N) \nonumber \\ & \quad + \sum _{i=1}^N \int\! \rmd \tilde x_i \,\rmd \tilde v_i\, \varphi(\tilde x_i - x_i) M_g^\varphi (\tilde x_i,\tilde v_i) {G} (Z_N^{i,(\tilde x_i,\tilde v_i)})\,.
\end{align}
We emphasize on the fact that the only difference w.r.t.\ Eq.~\eqref{gen} is that $M_{Z_N}^\varphi$ has been replaced by $M^\varphi_g$.

In proving the closeness between $Z_N(t)$ and $\Sigma_N(t)$ we find convenient to introduce the coupled process $Q_N(t) = (Z_N(t),\Sigma_N(t))$ given by the Markov process whose generator $\mc L_Q$ is defined in the following way. Denoting $Z_N=(X_N,V_N)$, $\Sigma_N=(Y_N,W_N)$, with $X_N=(x_1, \dots, x_N)$, $V_N=(v_1, \dots, v_N)$, $Y_N=(y_1, \dots, y_N)$, and $W_N=(w_1, \dots, w_N)$, and letting $G=G(Z_N,\Sigma_N)$ a test function, we set
\begin{align}
\label{genQ}
\mc L_Q{G}(Z_N,\Sigma_N) & = [(V_N\cdot\nabla_{X_N} + W_N\cdot\nabla_{Y_N}  - N){G}] (Z_N,\Sigma_N) \nonumber \\ & \quad + \sum _{i=1}^N \int\! \rmd \tilde x_i\, \rmd \tilde v_i\, \rmd \tilde y_i \, \rmd \tilde w_i \, \Phi_{x_i,y_i}(\tilde x_i,\tilde y_i) \nonumber  \\ & \qquad \qquad \times \mc M^\varphi (\tilde x_i,\tilde v_i;\tilde y_i,\tilde w_i) {G} (Z_N^{i,(\tilde x_i,\tilde v_i)},\Sigma_N^{i,(\tilde y_i,\tilde w_i)})\,,
\end{align}
where, for given $\tilde x,\tilde y\in \bb T^d$, $\mc M^\varphi(\tilde x,v;\tilde y,w)$ is a joint representation (to be fixed later on) of the Maxwellians $M_{Z_N}^\varphi(\tilde x,v)$ and $M_g^\varphi(\tilde y,w)$, and, for given $x,y\in \bb T^d$, $\Phi_{x,y}(\tilde x,\tilde y)$ is the joint representation of the probability densities $\varphi_x(\tilde x) = \varphi(\tilde x-x)$ and  $\varphi_y(\tilde y) = \varphi(\tilde y-y)$ defined as
\begin{equation}
\label{jointxy0}
\Phi_{x,y}(\tilde x,\tilde y) = \varphi_x(\tilde x) \delta(\tilde x -x - \tilde y +y)\,,
\end{equation}
where $\delta(x)$ denoted the Dirac measure on $\bb T^d$ centered in $x=0$. We remark that in particular, for any integrable function $J$ on $\bb T^d$,
\begin{equation}
\label{jointxy}
\int\! \rmd\tilde x\, \rmd\tilde y\, \Phi_{x,y}(\tilde x,\tilde y) J(\tilde x-\tilde y) = J(x-y)\,.
\end{equation}
In term of process, $(Z_N,\Sigma_N)$ performs jumps at random Poisson time of intensity $N$: at each jump time, $i$ is chosen uniformly and $(x_i,v_i,y_i,w_i) \to (x_i+\xi,y_i+\xi,\tilde v_i, \tilde w_i)$, where the common position jump $\xi$ is distributed according to $\varphi$ and $(\tilde v_i, \tilde w_i)$ according to the aforementioned joint representation of the two Maxwellians $M^\varphi_{Z_N}(\tilde x_i, \cdot)$  and $M^\varphi_g(\tilde y_i,\cdot)$ to be specified later.

Let now $R_N(t) = R_N(Z_N,\Sigma_N,t)$ be the law of $Q_N(t)$ and assume that, initially,
\[
R_N(0) = \delta(X_N-Y_N) \delta(V_N-W_N) f_0^{\otimes N}(X_N,V_N)\,.
\]
Then, setting
\[
I_N(t) := \int\! \rmd R_N(t) (|x_1-y_1|^2 + |v_1-w_1 |^2)
\]
and noticing that, as $R_N(t)$ is symmetric with respect to particle permutations,
\[
I_N(t) = \frac 1j \int\! \rmd R_N(t) \sum_{i=1}^j (|x_i-y_i|^2 + |v_i-w_i |^2) \qquad \forall\, j\in\{1,\ldots,N\}\,,
\]
the proof of Theorem \ref{teo:main} reduces to show that
\begin{equation}
\label{I_N}
I_N(t) \le \frac{C\Gamma_\varphi}N \exp (C\Gamma_\varphi) \,.
\end{equation}
Indeed, from the definition of the 2-Wasserstein distance it follows immediately that $\mc W_2\big(f_j^N(t),g(t)^{\otimes j}\big) \le \sqrt{j I_N(t)}$.

To prove Eq.~\eqref{I_N} we compute,
\begin{align*}
& \displaystyle \frac{\rmd}{\rmd t} I_N(t)  =   \int\! \rmd R_N(t)\, \mc L_Q (|x_1-y_1|^2 + |v_1-w_1 |^2) \\ & =  \int\! \rmd R_N(t)\, (v_1 \cdot \nabla_{x_1} + w_1 \cdot \nabla_{y_1})  |x_1-y_1|^2 \\ & \quad - N \int\! \rmd R_N(t)\,  ( |x_1-y_1|^2 + |v_1-w_1|^2) \\ &  \quad + \sum_{i=2}^N \int\! \rmd R_N(t)\, (|x_1-y_1|^2 + |v_1-w_1|^2) + \int\! \rmd R_N(t)\, |x_1-y_1|^2 \\ & \quad + \int\! \rmd R_N(t) \int\! \rmd \xi\, \varphi(\xi) \int\! \rmd \tilde v_1\, \rmd \tilde w_1\, \mc M^\varphi(x_1+\xi, \tilde v_1; y_1+\xi, \tilde w_1) |\tilde v_1- \tilde w_1|^2\,. 
\end{align*}
Here, the first two terms in the right-hand side arise from the stream part ($V_N\cdot\nabla_{X_N}G + W_N\cdot\nabla_{Y_N}G$) and the loss part ($-NG$) of the generator $\mc L_Q$, respectively. The loss part is largely compensated by the third term, which is the sum over all the particles but particle $1$, see the last term in the right-hand side of  Eq.~\eqref{genQ}. The last two terms are those arising from the remaining term $i=1$, separating the position  and velocity contributions and having used Eq.~\eqref{jointxy} in the former and Eq.~\eqref{jointxy0} in the latter.

We observe that the stream part is equal to
\[
2 \int\! \rmd R_N(t)\, (v_1-w_1) \cdot (x_1-y_1) \le \int\! \rmd R_N(t)\, (|x_1-y_1|^2+|v_1-w_1 |^2)\,,
\]
while, concerning the last term, we choose $\mc M^\varphi$ the optimal coupling that realizes the 2-Wasserstein distance between the marginals, whose square is given by (see, e.g., \cite{OP})
\[
\mc W_2\big(M_{Z_N}^\varphi(x,\cdot),M_g^\varphi(y,\cdot)\big)^2 = |u_N^\varphi(x)-u_g^\varphi(y)|^2 + d\Big|\sqrt{T_N^\varphi(x)} - \sqrt{T_g^\varphi(y)}\Big|^2\,.
\]

Collecting together the above formulas, we find that
\begin{equation}
\label{I<1}
\frac{\rmd}{\rmd t} I_N(t) \le I_N(t) + \int\! \rmd R_N(t)\, D(Z_N,\Sigma_N)\,,
\end{equation}
where
\begin{align}
\label{w2}
D(Z_N,\Sigma_N) & = \int\! \rmd \xi\, \varphi(\xi)\, |u_N^\varphi(x_1+\xi)-u_g^\varphi(y_1+\xi)|^2 \nonumber \\ & \quad + \int\! \rmd \xi\, \varphi(\xi)\, d\Big|\sqrt{T_N^\varphi(x_1+\xi)} - \sqrt{T_g^\varphi(y_1+\xi)}\Big|^2.
\end{align}

Our goal is to estimate from above the integral in the right-hand side of Eq.~\eqref{I<1} with a constant (independent of $N$)  multiple of $I_N(t)$ plus a small (order $1/N$) term. Then, Eq.~\eqref{I_N} will follow from Gr\"{o}nwall's inequality.

\subsection{Estimates}
\label{sec:3.3}

In estimating the function $D$ defined in Eq.~\eqref{w2}, it is convenient to replace $\varrho_g^\varphi$, $u_g^\varphi$, $T_g^\varphi$ with the fields $\tilde \varrho_N^\varphi$, $\tilde u_N^\varphi$, $\tilde T_N^\varphi$ given by
\[
\begin{split}
& \tilde \varrho_N^\varphi(x) = \frac 1N \sum_{j=1}^N \varphi (x-y_j)\,, \qquad  \tilde \varrho_N^\varphi \tilde u_N^\varphi(x) = \frac 1N \sum_{j=1}^N \varphi (x-y_j) w_j\,, \\ & \tilde \varrho_N^\varphi \tilde T_N^\varphi(x) = \frac 1{Nd} \sum_{j=1}^N \varphi (x-y_j) |w_j - \tilde u_N^\varphi(x)|^2 \,,
\end{split}
\]
i.e., the empirical fields constructed via the variables $Y_N=(y_1, \dots, y_N) $ and $W_N=(w_1, \dots, w_N)$, distributed independently according to $g(t)^{\otimes N}$. By the law of large numbers, the error due to this replacement in estimating the right-hand side of Eq.~\eqref{w2} will be shown to be small (order $1/N$).

More precisely, since from Eq.~\eqref{eq:TA},
\begin{align*}
\Big|\sqrt{T_N^\varphi(x)} - \sqrt{T_g^\varphi(y)}\Big|  & =  \frac{\big|T_N^\varphi(x) - T_g^\varphi(y)\big|}{\sqrt{T_N^\varphi(x)} +\sqrt{T_g^\varphi(y)}} \le \frac{\big|T_N^\varphi(x) - \tilde T_N^\varphi(y)\big|}{\sqrt{T_N^\varphi(x)} + \sqrt{A_t}} \\ & \quad + \frac{\big| \tilde T_N^\varphi(y) - T_g^\varphi(y)\big|}{\sqrt{A_t}}\,,
\end{align*}
we have,
\begin{equation}
\label{w2<}
D(Z_N,\Sigma_N) \le D_1(Z_N,\Sigma_N) + \mc E(\Sigma_N)\,,
\end{equation}
where
\begin{align}
\label{D1}
D_1(Z_N,\Sigma_N) & = \int\! \rmd \xi\, \varphi(\xi) \, 2 |u_N^\varphi(x_1+\xi) - \tilde u_N^\varphi(y_1+\xi)|^2 \nonumber \\ & \quad + \int\! \rmd \xi\, \varphi(\xi) \,2d \Bigg|\frac{T_N^\varphi(x_1+\xi) - \tilde T_N^\varphi(y_1+\xi)}{\sqrt{T_N^\varphi(x_1+\xi)} + \sqrt{A_t}}\Bigg|^2
\end{align}
and
\begin{align}
\label{E}
\mc E(\Sigma_N) & =  \int\! \rmd \xi\, \varphi(\xi) \, 2 |\tilde u_N^\varphi(y_1+\xi) - u_g^\varphi(y_1+\xi)|^2 \\ & \quad + \int\! \rmd \xi\, \varphi(\xi) \, 2d \frac{\big|  \tilde T_N^\varphi(y_1+\xi) - T_g^\varphi(y_1+\xi)\big|^2 }{A_t}\,.
\end{align}

\begin{lemma}
\label{lem:esti1}
Recall the definition of $\Gamma_\varphi$ in Eq.~\eqref{Cphi}. Then, for any $t\ge 0$,
\begin{align}
\label{D1<}
& D_1(Z_N,\Sigma_N) \nonumber \\ & \le C \Gamma_\varphi \left(1+ \frac 1N {\sum}_j |w_j|^4 \right)\bigg(\frac{|X_N-Y_N|^2+|V_N-W_N|^2}N +|x_1-y_1|^2\bigg)\,.
\end{align}
\end{lemma}

\begin{proof}
Before evaluating the difference between the empirical fields, we introduce the normalized weights,
\[
p_j = \frac{\varphi(x_1+\xi-x_j)}{\sum_k \varphi(x_1+\xi-x_k)}\,, \qquad q_j = \frac{\varphi(y_1+\xi-y_j)}{\sum_k \varphi(y_1+\xi-y_k)}.
\]
Recalling the definition of $C_\varphi$ in Eq.~\eqref{Cphi}, we have,
\begin{equation}
\label{pq1}
\max\{p_j; q_j\} \le \frac{C_\varphi\|\varphi\|_\infty}N\,.
\end{equation}
Moreover, since
\[
\begin{split}
p_j - q_j  & = \frac{\varphi(x_1+\xi-x_j) - \varphi(y_1+\xi-y_j)}{\sum_k \varphi(x_1+\xi-x_k)} \\ & \quad + \varphi(y_1+\xi-y_j) \frac{\sum_k [\varphi(x_1+\xi-x_k) - \varphi(y_1+\xi-y_k)]}{\sum_k \varphi(x_1+\xi-x_k)\sum_k \varphi(y_1+\xi-y_k)},
\end{split} 
\]
it also follows that
\begin{align}
|p_j-q_j| & \le \frac{C_\varphi\|\nabla\varphi\|_\infty}N (|x_1-y_1| + |x_j-y_j|) \nonumber \\ & \quad + \frac{C_\varphi^2\|\varphi\|_\infty \|\nabla \varphi\|_\infty}{N^2} \sum_k (|x_1-y_1| + |x_k-y_k|) \nonumber \\ & \le \frac{C_\varphi\|\nabla\varphi\|_\infty}N \big(1+C_\varphi\|\varphi\|_\infty\big) |x_1-y_1| + \frac{C_\varphi\|\nabla\varphi\|_\infty}N |x_j-y_j| \nonumber \\ \label{pq2} & \quad + \frac{C_\varphi^2\|\varphi\|_\infty \|\nabla \varphi\|_\infty}N \frac{|X_N-Y_N|}{\sqrt N}\,,
\end{align}
where in the last bound we used that, by the Cauchy-Schwarz inequality, $\sum_k |x_k-y_k|\le \sqrt N |X_N-Y_N|$.

Regarding the first term in the right-hand side of Eq.~\eqref{D1}, we notice that
\[
|u_N^\varphi(x_1+\xi) - \tilde u_N^\varphi(y_1+\xi)| \le U_1 + U_2\,,
\]
where, by Eqs.~\eqref{pq1} and \eqref{pq2},
\begin{align*}
U_1 & = \sum_j p_j |v_j-w_j|  \le \frac{C_\varphi\| \varphi\|_\infty} N \sum_j|v_j-w_j| \,, 
\\ U_2 & = \sum_j|p_j-q_j| |w_j| \le  \frac{C_\varphi\|\nabla\varphi\|_\infty}N \big(1+C_\varphi\|\varphi\|_\infty\big) |x_1-y_1|\sum_j |w_j| \\ & \quad + \frac{C_\varphi\|\nabla\varphi\|_\infty}N \sum_j |x_j-y_j| |w_j| +  C_\varphi^2\|\varphi\|_\infty \|\nabla \varphi\|_\infty \frac{|X_N-Y_N|}{\sqrt N} \frac 1N \sum_j|w_j|\,.
\end{align*}
Therefore, again by the Cauchy-Schwarz inequality, 
\begin{align}
\label{eq:uu}
& |u_N^\varphi(x_1+\xi) - \tilde u_N^\varphi(y_1+\xi)|^2 \le 2 U_1^2 + 2U_2^2 \le  2C_\varphi^2 \|\varphi\|_\infty^2 \frac{|V_N-W_N|^2}N \nonumber \\ & \quad + C (1+ C_\varphi^2\|\varphi\|_\infty^2) C_\varphi^2 \|\nabla \varphi\|_\infty^2 \frac{|W_N|^2}N \bigg(\frac{|X_N-Y_N|^2}N +|x_1-y_1|^2\bigg)\,.
\end{align}
The estimate of the second term in the right-hand side of Eq.~\eqref{D1} is more tricky and requires some effort. We first notice that
\[
\Bigg| \frac{T_N^\varphi(x_1+\xi) - \tilde T_N^\varphi(y_1+\xi)}{\sqrt{T_N^\varphi(x_1+\xi)} + \sqrt{A_t}}\Bigg| \le T_1 + T_2\,,
\]
where (omitting the explicit dependence on $x_1+\xi$ and $y_1+\xi$)
\[
T_1 = \frac 1d \sum_j p_j \Bigg|\frac{|v_j-u_N^\varphi|^2 - |w_j - \tilde u_N^\varphi|^2}{\sqrt{T_N^\varphi} + \sqrt{A_t}}\Bigg|\,, \qquad T_2 = \frac 1d \sum_j |p_j-q_j| \frac{|w_j - \tilde u_N^\varphi|^2}{\sqrt{A_t}}\,.
\]
Now, 
\[
\begin{split}
T_1 & = \frac 1d\sum_j p_j \Bigg|\frac{(v_j - u_N^\varphi - w_j + \tilde u_N^\varphi)\cdot (v_j - u_N^\varphi + w_j - \tilde u_N^\varphi)}{\sqrt{T_N^\varphi} + \sqrt{A_t}}\Bigg| \le T_{1,1} + T_{1,2}\,,
\end{split}
\]
with
\[
\begin{split}
T_{1,1} & = \frac 1d \sum_j p_j \frac{(|v_j -w_j| + |u_N^\varphi -\tilde u_N^\varphi|) |v_j-u_N^\varphi|}{\sqrt{T_N^\varphi} + \sqrt{A_t}}\,, \\
T_{1,2} & = \frac 1{d\sqrt{A_t}}\sum_j p_j (|v_j -w_j| + |u_N^\varphi -\tilde u_N^\varphi|) |w_j-\tilde u_N^\varphi|\,.
\end{split}
\]
By the Cauchy-Schwartz inequality with respect to the weights $\{p_j\}$ and Eq.~\eqref{pq1},
\begin{align*}
T_{1,1} & \le \frac 1d \sqrt{{\sum}_j p_j (|v_j -w_j| + |u_N^\varphi -\tilde u_N^\varphi|)^2}\, \frac{\sqrt{T_N^\varphi}}{\sqrt{T_N^\varphi} + \sqrt{A_t}} \\ & \le \frac 1d \sqrt{ \frac{2C_\varphi\|\varphi\|_\infty}N |V_N-W_N|^2 + 2|u_N^\varphi -\tilde u_N^\varphi|^2}
\end{align*}
and
\begin{align*}
& T_{1,2} \le \frac 1{d\sqrt{A_t}} \sqrt{{\sum}_j p_j (|v_j -w_j| + |u_N^\varphi -\tilde u_N^\varphi|)^2}\, \sqrt{{\sum}_j p_j |w_j - \tilde u_N^\varphi|^2} \\ & \le \frac 1{d\sqrt{A_t}} \sqrt{\frac{2C_\varphi\|\varphi\|_\infty}N |V_N-W_N|^2 + 2|u_N^\varphi -\tilde u_N^\varphi|^2}\, \sqrt{\frac{C_\varphi\|\varphi\|_\infty}N {\sum}_j |w_j-\tilde u_N^\varphi|^2}\,.
\end{align*}
On the other hand, in view of Eq.~\eqref{pq2} and by the Cauchy-Schwarz inequality, 
\begin{align*}
T_2 & \le \frac{C_\varphi\|\nabla\varphi\|_\infty}{d\sqrt{A_t} N} \sum_j \Big[\big(1+C_\varphi\|\varphi\|_\infty\big) |x_1-y_1|  +  |x_j-y_j|\Big] |w_j - \tilde u_N^\varphi|^2   \\ & \quad + \frac{C_\varphi^2\|\varphi\|_\infty \|\nabla \varphi\|_\infty}{d\sqrt{A_t}} \frac{|X_N-Y_N|}{\sqrt N} \frac 1N \sum_j|w_j - \tilde u_N^\varphi|^2 \\ & \le \frac{C_\varphi\|\nabla\varphi\|_\infty}{d\sqrt{A_t}}\left[ \big(1+C_\varphi\|\varphi\|_\infty\big) |x_1-y_1| + \frac{|X_N-Y_N|}{\sqrt N}\right] \sqrt{\frac 1N {\sum}_j|w_j - \tilde u_N^\varphi|^4}\\ & \quad + \frac{C_\varphi^2\|\varphi\|_\infty \|\nabla \varphi\|_\infty}{d\sqrt{A_t}} \frac{|X_N-Y_N|}{\sqrt N} \frac 1N \sum_j|w_j - \tilde u_N^\varphi|^2\,.
\end{align*}
We finally observe that, as $|\tilde u_N^\varphi| \le C_\varphi\|\varphi\|_\infty |W_N|/\sqrt N$,
\begin{align}
\label{eq:w4}
\frac 1N \sum_j|w_j - \tilde u_N^\varphi|^2 & \le 2(1+C_\varphi^2\|\varphi\|_\infty^2) \frac{|W_N|^2}N\,, \nonumber \\ \frac 1N \sum_j|w_j - \tilde u_N^\varphi|^4 & \le \frac 4N {\sum}_j |w_j|^4 + 4C_\varphi^4 \|\varphi\|_\infty^4 \frac{|W_N|^4}{N^2}\,.
\end{align}
From the above estimates on $T_{1,1}$, $T_{1,2}$, $T_2$, and inequalities Eqs.~\eqref{eq:uu} and \eqref{eq:w4}, we conclude that
\begin{align}
\label{eq:TT}
&\displaystyle \left|\frac{T_N^\varphi(x_1+\xi) - \tilde T_N^\varphi(y_1+\xi)}{\sqrt{T_N^\varphi(x_1+\xi)} + \sqrt{A_t}} \right|^2 \le  C\Gamma_\varphi \left(\frac{|W_N|^4}{N^2} + \frac 1N {\sum}_j |w_j|^4\right) |x_1-y_1|^2 \nonumber \\ & \quad \qquad + C\Gamma_\varphi \left(\frac{|W_N|^2}N + \frac{|W_N|^4}{N^2} + \frac 1N {\sum}_j |w_j|^4\right) \frac{|X_N-Y_N|^2}N\nonumber \\ & \quad \qquad +  C\Gamma_\varphi \left(1+\frac{|W_N|^2}N\right) \frac{|V_N-W_N|^2}N\,.
\end{align}
Since $\frac 1N |W_N|^2 \le \sqrt{\frac 1N {\sum}_j |w_j|^4}$, Eq.~\eqref{D1<} follows from Eqs.~\eqref{eq:uu} and \eqref{eq:TT} (recall $\int\!\rmd \xi\,\varphi(\xi) =1$).
\qed\end{proof}

\begin{lemma}
\label{lem:esti2}
Recall the definition of $\Gamma_\varphi$ in Eq.~\eqref{Cphi}. Then, for any $t\ge 0$,
\begin{equation}
\label{eq:w8}
\int\!\rmd R_N(t) \, \frac 1N \sum_j \big(|v_j|^8 + |w_j|^{16}\big) \le C \exp(C\Gamma_\varphi)\,.
\end{equation}
\end{lemma}

\begin{proof}
From the estimate on $\mc N_q(g)$ in Eq.~\eqref{eq:utK} and the symmetry of $F_N(t)$ we have,
\[
\int\!\rmd R_N(t) \, \frac 1N \sum_j \big(|v_j|^8 + |w_j|^{16}\big)  = \int\!\rmd Z_N\, F_N(Z_N,t) \, |v_1|^8 + C\,,
\]
so that we only need an upper bound on the first term in the right-hand side. To this purpose, recalling the explicit expression of the generator Eq.~\eqref{gen}, we compute,
\begin{align}
\label{stv}
& \frac{\rmd}{\rmd t} \int\!\rmd  F_N(t) \, |v_1|^8   \nonumber \\ & \qquad = - \int\!\rmd Z_N\, F_N(Z_N,t)\, |v_1|^8+ \int\!\rmd Z_N\, F_N(Z_N,t)\int\!\rmd \xi\,\varphi(\xi)  \nonumber \\ & \qquad\qquad \times\int\! \rmd \tilde v_1\, M_{Z_N}^\varphi (x_1+\xi,\tilde v_1) |\tilde v_1|^8 \nonumber\\ & \qquad =  - \int\!\rmd Z_N\, F_N(Z_N,t)\, |v_1|^8 + \int\!\rmd Z_N\, F_N(Z_N,t)\int\!\rmd \xi\,\varphi(\xi)  \nonumber \\ & \qquad\qquad \times \int\! \rmd \eta \, M_{0,1}(\eta) \, \big|u_N^\varphi(x_1+\xi) - \sqrt{T_N^\varphi(x_1+\xi)}\,\eta\big|^8\, ,
\end{align}
where $ M_{0,1}$ is the Gaussian centered in $0$ with unitary variance. Since, in view of Eq.~\eqref{pq1},
\[
|u_N^\varphi(x_1+\xi)| \le \frac{C_\varphi \|\varphi\|_\infty}N \sum_j |v_j| \,, \qquad T_N^\varphi(x_1+\xi) \le \frac{C_\varphi \|\varphi\|_\infty}N \sum_j |v_j|^2\,,
\]
the Gaussian integral in the right-hand side of Eq.~\eqref{stv} is bounded by a constant multiple of $(C_\varphi \|\varphi\|_\infty)^8 \frac 1N \sum_j |v_j|^8$, so that, by using again the symmetry of $F_N(t)$,
\[
\frac{\rmd}{\rmd t} \int\!\rmd  F_N(t) \, |v_1|^8  \le C (C_\varphi \|\varphi\|_\infty)^8\int\!\rmd  F_N(t) \, |v_1|^8\,,
\]
from which the lemma follows by Gr\"{o}nwall's inequality and the assumption on the initial distribution function $f_0$ given in Eq.~\eqref{eq:f0}.
\qed\end{proof}

\subsection{Conclusion (proof of Eq.~(\ref{I_N}))}
\label{sec:3.4}

Our task is to estimate from above the integral in the right-hand side of Eq.~\eqref{I<1} with a constant multiple of $I_N(t)$ plus a small (order $1/N$) term. To this end, we use Eq.~\eqref{w2<} and decompose 
\begin{equation}
\label{D1st0}
\int\! \rmd R_N(t)\, D(Z_N,\Sigma_N) = K_1 + K_2\,,
\end{equation}
with 
\[
K_1 = \int\! \rmd R_N(t)\, D_1(Z_N,\Sigma_N)\,, \qquad  K_2 = \int\! \rmd g(t)^{\otimes N}\, \mc E(\Sigma_N)\,.
\]
Recalling the definition Eq.~\eqref{E} of $\mc E$, from the law of large numbers we have that
\begin{equation}
\label{K3}
K_2 \le \frac CN\,.
\end{equation}
Now, we want to use Lemma \ref {lem:esti1} to bound $K_1$ by means of $I_N(t)$. This is not immediate, due to the factor $\big(1+ \frac 1N {\sum}_j |w_j|^4 \big)$ appearing in the right-hand side of Eq.~\eqref{D1<}. The strategy is to decompose the domain of integration, by introducing  the ``good set'' 
\[
\mc G_a = \left\{(Z_N,\Sigma_N) \colon {\sum}_j |w_j|^4 \le Na \right\}, \quad a>0\,,
\]
where the right-hand side of Eq.~\eqref{D1<} is under control, and show that for $a$ sufficiently large the contribution of the integration outside $\mc G_a$ is order $1/N$(due to the law of large numbers). With this in mind, we decompose
\begin{equation}
\label{D1st}
K_1 = K_{1,1} + K_{1,2}\,,
\end{equation}
with 
\[
K_{1,1} = \int_{\mc G_a}\! \rmd R_N(t)\, D_1(Z_N,\Sigma_N)\,,\qquad K_{1,2} = \int_{\mc G_a^\complement}\! \rmd R_N(t)\, D_1(Z_N,\Sigma_N)\,.
\]
In view of Eq.~\eqref{D1<}, if $C_t(a) = C \Gamma_\varphi (1+a)$ then, for any $ (Z_N,\Sigma_N) \in \mc G_a$,
\[
D_1(Z_N,\Sigma_N) \le C_t(a) \bigg(\frac{|X_N-Y_N|^2+|V_N-W_N|^2}N +|x_1-y_1|^2\bigg)
\]
so that, noticing that the law $R_N(t)$ is symmetric,
\begin{align}
\label{K1}
K_{1,1} & \le C_t(a) \int\! \rmd R_N(t)\, \bigg(\frac{|X_N-Y_N|^2+|V_N-W_N|^2}N +|x_1-y_1|^2\bigg) \nonumber \\ & \le 2 C_t(a) I_N(t)\,.
\end{align}

In estimating $K_{1,2}$, we first observe that
\begin{align}
\label{K2b}
K_{1,2} & \le \Bigg(\int\! \rmd R_N(t)\, D_1(Z_N,\Sigma_N)^2 \Bigg)^{1/2} \Bigg(\int_{{\sum}_j |w_j|^4 > Na}\! \rmd g(t)^{\otimes N}\, \Bigg)^{1/2} \nonumber \\ & \le C\Gamma_\varphi \exp(C\Gamma_\varphi)  \Bigg(\int_{{\sum}_j |w_j|^4 > Na}\! \rmd g(t)^{\otimes N} \Bigg)^{1/2} \,,
\end{align}
where, in the last inequality, we first used that the square of the right-hand side in Eq.~\eqref{D1<} is bounded by a constant multiple of $\Gamma_\varphi\frac 1N \sum_j \big(|v_j|^8 + |w_j|^{16}\big)$, and then we applied Eq.~\eqref{eq:w8}. We now show that, by the law of large numbers, if $a$ is large enough then the integral in the right-hand side of Eq.~\eqref{K2b} is vanishing as $N\to +\infty$. More precisely, from the estimate on $\mc N_q(g)$ in Eq.~\eqref{eq:utK}, there is $M=M(t,f_0)$ such that $\int\!\rmd y\, \rmd w\, g(y,w,t) |w|^4 \le M$. Therefore, letting $\xi_N= \frac 1N\sum_j|w_j|^4$ and $\bb E(\xi_N) = \int\! \rmd g(t)^{\otimes N} \xi_N$, if $a>2M$ we have
\[
\int_{{\sum}_j |w_j|^4 > Na}\! \rmd g(t)^{\otimes N} \le \int_{|\xi_N-\bb E(\xi_N)|>M}\! \rmd g(t)^{\otimes N} \,.
\]
Therefore, by the law of large numbers (i.e., Chebyshev's inequality), 
\[
\int_{{\sum}_j |w_j|^4 > Na}\! \rmd g(t)^{\otimes N} \le \frac CN\,,
\]
so that
\begin{equation}
\label{K2}
K_{1,2} \le \frac{C\Gamma_\varphi}N \exp(C\Gamma_\varphi) \,.
\end{equation}

In view of Eqs.~\eqref{I<1}, \eqref{D1st0}, \eqref{D1st}, \eqref{K1}, \eqref{K2}, and \eqref{K3}, we conclude that, for any $0<s\le t$,
\begin{align}
\label{I<11}
\frac{\rmd}{\rmd s} I_N(s)  & \le [1+2C_t(a)] I_N(s) + \frac{C\Gamma_\varphi}N \exp(C\Gamma_\varphi)\,,
\end{align}
which implies Eq.~\eqref{I_N} by Gr\"{o}nwall's inequality.

\section{Lipschitz estimates and removal of the cut-off}
\label{sec:4}

In this section we prove Theorem \ref{thm:fgd}. A preliminary result is the following lemma, where we provide $L^\infty$ bounds on the spatial derivatives of the solutions to either the BGK equation, Eq.~\eqref{eq:bgk}, or its regularized version, Eq.~\eqref{eq:kin}. We emphasize that, in the latter case, these estimates do not depend on the smearing function $\varphi$.

\begin{lemma}
\label{lem:dgf}
Under the hypothesis of Theorem \ref{thm:fgd}, for any $t\ge 0$,
\begin{equation}
\label{dg}
\mc N_q(|\nabla_x f(t)|) + \mc N_q(|\nabla_x g(t)|) \le C \,,
\end{equation}
where $f(t)$ [resp.~$g(t)$] is the solution to the BGK [resp.~regularized BGK] equation with initial condition $f(0) = g(0) = f_0$ given by Proposition \ref{prop:bgk} [resp.~Proposition \ref{prop:stim_uT}].
\end{lemma}

\begin{proof}
We prove the claim for the solution to the BGK equation, the case of the regularized BGK equation can be treated in the same way.

By differentiating Eq.~\eqref{eq:bgk} we have,
\[
(\partial_t + v\cdot \nabla_x + 1) \nabla_x f = \varrho_f M_f Q_f\,,
\]
with 
\[
Q_f = \frac{\nabla_x \varrho_f}{\varrho_f} + \frac{(D_x u_f)^T (v-u_f)}{T_f} + \bigg(\frac{|v-u_f|^2}{2T_f^2}  - \frac{d}{4\pi T_f} \bigg) \nabla_x T_f\,.
\]
Therefore, by Duhamel formula,
\begin{equation}
\label{eq:3}
\nabla_x f(x,v,t) = \rme^{-t}\nabla_x f_0(x-vt,v) + \int_0^t\!\rmd s\, \rme^{-(t-s)} (\varrho_f M_f Q_f)(x-v(t-s),v,s)\,.
\end{equation}
To estimate the $\mc N_q$-norm of the integral in the right-hand side of Eq.~\eqref{eq:3} we first observe that
\[
\begin{split}
(1+|v|^q) |v-u_f|^j M_f & \le C (1+|v-u_f|^q + |u_f|^q) |v-u_f|^j M_f \\ & \le C T_f^{(j-d)/2}(1 +T_f^{q/2} + |u_f|^q)\qquad \forall\, j=0,1,2\,.
\end{split}
\]
Moreover,
\begin{align*}
|\nabla_x \varrho_f| & \le \int\!\rmd v\, |\nabla_x f| \le C \mc N_q(|\nabla_x f|)\,, \\ |D_x u_f| & \le \frac{|u_f| |\nabla_x\varrho_f|}{\varrho_f} + \frac{1}{\varrho_f}\int\!\rmd v\, |\nabla_x f| \, |v| \le C \frac{(1+|u_f|) \mc N_q(|\nabla_xf|)}{\varrho_f}\,, \\ |\nabla_x T_f|  & \le \frac{T_f |\nabla_x\varrho_f|}{\varrho_f} + \frac{1}{\varrho_f}\int\!\rmd v\, |\nabla_x f| \, |v-u_f|^2 \\ & \le C \frac{(T_f+1+|u_f|^2) \mc N_q(|\nabla_xf|)}{\varrho_f}\,,
\end{align*}
where we used that if $q>d+2$ then
\[
\int\! \rmd v\, |\nabla_x f| |v|^j = \int\! \rmd v\, |\nabla_x f| |v|^j \frac{1+|v|^q}{1+|v|^q} \le C \mc N_q(|\nabla_x f|) \qquad \forall\, j=0,1,2\,.
\]
Therefore, in view of Eqs.~\eqref{eq:utf} and \eqref{eq:rT>}, from the above estimates we deduce that $\mc N_q(\varrho_f M_f Q_f) \le C \mc N_q(|\nabla_xf|)$. The estimate on $ \mc N_q(|\nabla_xf|)$ then follows from Eq.~\eqref{eq:3} and Gr\"{o}nwall's inequality.
\qed\end{proof}

\noindent\textit{Proof of Theorem \ref{thm:fgd}}\hspace{2truept} 
We introduce the shorten notation $\varrho^\eps$, $u^\eps$, $T^\eps$ to denote the smeared local fields defined as in Eqs.~\eqref{eq:rphi}, \eqref{eq:uphi}, and \eqref{eq:tphi} with $\varphi_\eps$ in place of $\varphi$.

From Duhamel formula,
\[
f(x,v,t)- g^\eps(x,v,t) = \int_0^t\!\rmd s\, \rme^{-(t-s)} (\varrho_f M_f - \varrho^\eps M_{g^\eps})(x-v(t-s),v,s)\,,
\]
so that, setting
\[
D(t) = \int\!\rmd x\, \rmd v\, (1+ |v|^2) \, |f(x,v,t) - g^\eps(x,v,t)|\,,
\]
we have (after the change of variable $x \to x+v(t-s)$ on $\bb T^d$)
\begin{equation}
\label{eq:4}
D(t) \le \int_0^t\!\rmd s \int\!\rmd x\, \rmd v\,  (1+ |v|^2)  \big| (\varrho_f M_f - \varrho^\eps M_{g^\eps})(x,v,s) \big|\,.
\end{equation}

To estimate the right-hand side in Eq.~\eqref{eq:4} we argue as in \cite{PP}.  We set, for $\lambda\in [0,1]$,
\[
(\varrho_\lambda, u_\lambda, T_\lambda) = \lambda (\varrho_f, u_f, T_f) + (1-\lambda) (\varrho^\eps, u^\eps, T^\eps)
\]
and let $M_\lambda(v) = M_{u_\lambda,T_\lambda}(v)$, so that
\begin{align*}
& \int\!\rmd v\, (1+ |v|^2)\,  |\varrho_f M_f - \varrho^\eps M_{g^\eps}| \le \int_0^1\!\rmd \lambda\, \int\!\rmd v\, (1+ |v|^2)\, \bigg\{ |\varrho_f - \varrho^\eps| M_\lambda \\ & \hskip 3cm  + \varrho_\lambda |\nabla_u M_\lambda| \, |u_f - u^\eps| + \varrho_\lambda \bigg|\frac{\partial M_\lambda}{\partial T}\bigg| \, |T_f - T_{g^\eps}| \bigg\}  \\ & \hskip 2 cm \le C( |\varrho_f - \varrho^\eps| + |u_f - u^\eps| +  |T_f - T_{g^\eps}|)\,,
\end{align*}
where, in obtaining the last inequality, we first used that
\[
\begin{split}
\int\!\rmd v\, (1+ |v|^2)\, M_\lambda & \le 1 + |u_\lambda|^2 + T_\lambda\,, \\
\int\!\rmd v\, (1+ |v|^2)\,  |\nabla_u M_\lambda| & \le C \frac{1 + |u_\lambda|^2 + T_\lambda}{\sqrt{T_\lambda}}\,, \\ \int\!\rmd v\, (1+ |v|^2)\, \bigg|\frac{\partial M_\lambda}{\partial T}\bigg| & \le C \frac{1 + |u_\lambda|^2 + T_\lambda}{T_\lambda}\,,
\end{split}
\]
and then applied the lower and upper bounds on the hydrodynamical fields  given in Propositions \ref{prop:bgk} and \ref{prop:stim_uT}. 
Now, again from these propositions,
\[
\begin{split}
|u_f - u^\eps| & \le C \varrho_f |u_f - u^\eps| \le C (|\varrho_f u_f - \varrho^\eps u^\eps| + |u^\eps| \, |\varrho_f - \varrho^\eps|) \\ & \le C \int\!\rmd v\, (1+ |v|^2)\, |f-g^\eps| + C |\varrho_f - \varrho^\eps|\,, \\ 
|T_f - T^\eps| & \le C \varrho_f |T_f - T^\eps| \le C (|\varrho_f T_f - \varrho^\eps T^\eps| + |T^\eps| \, |\varrho_f - \varrho^\eps|) \\ & \le C \int\!\rmd v\, (1+ |v|^2)\, |f-g^\eps| + C |\varrho_f - \varrho^\eps|\,.
\end{split}
\]
Finally,
\[
\begin{split}
& |\varrho_f - \varrho^\eps|(x) \le \int\!\rmd v\,\Big| f(x,v,t) - \int\!\rmd y\, \varphi_\eps(x-y) g^\eps(y,v,t)\Big| \\ & \qquad \le \int\!\rmd v\, (1+ |v|^2)\, |f-g^\eps| + \int\!\rmd y \, \rmd v\, \varphi_\eps(x-y) |g^\eps(y,v,t) - g^\eps(x,v,t)| \\ & \qquad  \le \int\!\rmd v\, (1+ |v|^2)\, |f-g^\eps| + C\eps\,,
\end{split}
\]
where we used Eq.~\eqref{stim2} and Lemma \ref{lem:dgf} in the last inequality.

Inserting the above bounds in Eq.~\eqref{eq:4} we finally have,
\[
D(t) \le C \int_0^t\!\rmd s\,  D(s) + C \eps\,,
\]
which implies Eq.~\eqref{eq:f-g} by Gr\"{o}nwall's inequality.
\qed

\begin{remark}
\label{rem:1}
Note that  the convergence part of the particle approximation is carried out by using a weak topology. Actually, this is natural since such a proof is based on the law of large numbers. In contrast, in removing the cut-off we used a weighted $L^1$ topology. A direct use of the weak topology could be possible also in this part of the proof, but in this context it is much less natural, being the proof more complicate and the result weaker.
\end{remark}

\section{Concluding remarks}
\label{sec:5}

Let us consider, for the moment, the non-physical particle dynamics introduced in the present paper as really describing the behavior of the microscopic world. Then, it makes sense to exploit the scaling and the regime for which the kinetic picture given by the BGK model is appropriate. Proceeding as for the most popular kinetic equations, let $(X_N,V_N)  \in (\bb T^d)^N\times (\bb R^d)^N$ be the macroscopic variables. The evolution of the microscopic system takes place in $ \bb T_\eps^d $, the $d$-dimensional torus of side $\eps^{-1}$, where $\eps$ is a scale parameter. In other words, the microscopic variables are
\[
(\eps^{-1} X_N, V_N, \eps^{-1}t)
\]
(velocities are unscaled), and the time evolution of the law $\bar F_N$ of the microscopic process is given by the Fokker-Planck equation,
\begin{align*}
& (\partial_{\eps^{-1} t} + V_N\cdot\nabla_{\eps^{-1} X_N} )\bar F_N(\eps^{-1}X_N,V_N, \eps^{-1}t)  = - \gamma_N N \bar F_N (\eps^{-1} X_N, V_N, \eps^{-1}t ) \\ & \qquad + \gamma_N \sum _{i=1}^N \int\! \rmd \tilde x_i \int\! \rmd \tilde v_i\, \eps^{-d} \varphi (\eps^{-1} (x_i-\tilde x_i)) M_{(\eps^{-1} X_N^{i,\tilde x_i}, V_N^{i,\tilde v_i})}^\varphi (\eps^{-1} x_i, v_i) \\ & \qquad\qquad\qquad \times \bar F_N (\eps^{-1} X_N^{i,\tilde x_i}, V_N^{i,\tilde v_i}, \eps^{-1}t)\,, 
\end{align*}
where $\gamma_N$ modulates the intensity of the jump process suitably and $\varphi$ is not scaled. Note that $\tilde x_i$ in the above formula is a macroscopic variable which belongs to the unitary torus.

Actually $\varphi$ describes the interaction. A possible choice is the characteristic function of the unitary sphere or a smooth version of it. This means that only the particles at distance at most $1$ from a given particle determine its random jumps. 

Denoting by
\[
F_N(X_N,V_N,t) = \eps^{-dN} \bar F_N (\eps^{-1}X_N,V_N, \eps^{-1}t)
\]
the law expressed in the macro-variables, we arrive to
\begin{align}
\label{FP}
& (\partial_{ t} + V_N\cdot\nabla_{ X_N}) F_N(X_N,V_N, t) = - \frac {\gamma_N}{\eps} N  F_N ( X_N, V_N, t ) \nonumber \\ & \quad + \frac {\gamma_N}{\eps}  \sum _{i=1}^N  \int\! \rmd \tilde x_i \int\! \rmd \tilde v_i\,  \varphi_\eps (x_i-\tilde x_i) 
M_{(X_N^{i,\tilde x_i}, V_N^{i,\tilde v_i})}^{\varphi_\eps}  (x_i, v_i) {F_N} (X_N^{i,\tilde x_i}, V_N^{i,\tilde v_i}, t)\,. 
\end{align}
Here, we used that
\[
M_{(\eps^{-1}X_N,V_N)}^\varphi (\eps^{-1} x_i, v_i)=M_{Z_N}^{\varphi_\eps}  (x_i, v_i)\,,
\]
with $\varphi_\eps (x) = \eps^{-d} \varphi (x/\eps)$, as follows by a direct inspection. Indeed, it follows that
$$
u_N^\varphi ( \eps^{-1} x_i)=u_N^{\varphi_\eps}  (  x_i) \quad T_N^\varphi ( \eps^{-1} x_i)=T_N^{\varphi_\eps}  (  x_i)
$$
with the convention that the left-hand side is computed via $ \eps^{-1}  X_N, V_N $ and the right hand side via $X_N,V_N$.

Next, we assume the microscopic density $O(1)$, thus $N=\eps^{-d}$ (hydrodynamical density). We recall here that the hydrodynamic limit consists in scaling space and time only, the evolution for the hydrodynamical fields being obtained via the quick local thermalization toward the local equilibrium. In contrast, the kinetic description requires a suitable modification of the dynamics to moderate the number of the interactions per unit (macroscopic) time. In the present context, we do this by rescaling $\gamma_N = \eps$.  

In conclusion, we recover the particle dynamics we have considered, but the condition $\eps = N^{-1/d}$ is too severe for our approach, as we need $\eps \approx (\log N)^{-a}$ for some positive $a$.

Moreover, in contrast with the setting discussed in the present section, in our derivation of the BGK equation we have assumed that $\varphi$ is strictly positive, excluding the case of the characteristic function of a unitary ball. Therefore, we face now a new potential divergence related to a possible rarefaction in a given box. However, this issue could probably be handled via extra probabilistic estimates.

We finally mention that if we consider the Fokker-Planck equation \eqref{FP} with $\gamma_N$ replaced by $ \frac 1N \sum_j \varphi (x_i-x_j) $ we expect, at least formally, to recover the BGK equation with $\lambda=\varrho(x)$.

\appendix
\normalsize

\section{Proof of Proposition \ref{prop:stim_uT}}
\label{app:a}

We first observe that the smeared fields, defined in Eqs.~\eqref{eq:rphi}, \eqref{eq:uphi}, and \eqref{eq:tphi}, coincide with the usual hydrodynamical fields associated to the smeared distribution function $g^\varphi(x,v) := \int\!\rmd y\, \varphi(x-y) g(y,v)$, i.e.,
\[
\varrho_g^\varphi = \varrho_{g^\varphi}\,, \quad \varrho_g^\varphi u_g^\varphi =  \varrho_{g^\varphi} u_{g^\varphi}\,, \quad \varrho_g^\varphi T_g^\varphi =  \varrho_{g^\varphi} T_{g^\varphi}\,.
\]
Therefore, according to \cite[Proposition 2.1]{PP}, we find the following pointwise estimates for $\rho_g^\varphi$, $u_g^\varphi$, and   $T_g^\varphi$:
\begin{align*}
& \text{(i) } \frac {\rho^\varphi }{ (T^\varphi )^{d/2}} \leq C N_0( g^\varphi)\,; \\ & \text{(ii) } \rho^\varphi ( T^\varphi +(u^\varphi )^2)^{ (q-d)/2} \leq C_q N_q (g^\varphi )\text{ either for $q> d+2$ or for $0\leq q <d$}\,; \\
& \text{(iii) } \frac{\rho^\varphi |u^\varphi|^{d+q}}{[ T^\varphi +(u^\varphi )^2) T^\varphi ]^{d/2} } \text{ for $q>1$}\,.
\end{align*}
Above, $C,C_q$ are constants independent of $\varphi$ and
\[
N_q(f) = \sup_v |v|^q f(v)\,, \quad    q\geq 0\,,
\]
for a given positive function $f$. 

As a consequence, following \cite{PP}, we infer
\[
\sup_v  |v|^q M_g^{\varphi} (x,v) \le C_q N_q( g^\varphi)
\]
and hence, writing the equation for $g$ in mild form and recallig Eq. \eqref{eq:Nn}, we obtain
\[
\mc N_q(g(t)) \le \mc N_q(f_0)+C_q\int_0^t\!\rmd s\,  \mc N_q( g^\varphi (s))\,.
\]
The a priori bound $\mc N_q(g(t)) \le  \mc N_q(f_0) \exp(C_q t)$ follows by  the obvious inequality $ N_q( g^\varphi (s)) \leq N_q( g (s))$ and the Gr\"{o}nwall's lemma.

Provided this estimate, by arguing exactly as in the proof of \cite[Theorem 3.1]{PP}, we construct the solution $g(t)$ by establishing the Lipschitz continuity of the operator $g \to \rho_g^\varphi M_g^{\varphi} -g$ in $L^1( (1+v^2) \rmd x \rmd v)$ and using the standard iteration scheme. Moreover, exactly as in \cite{PP}, the bounds \eqref{eq:utK}, \eqref{stimrho}, and \eqref{eq:TA} follow from this construction and the previous a priori estimates.

\end{document}